\documentclass[pra,twocolumn,showpacs,superscriptaddress]{revtex4-1}
\usepackage{graphicx}
\usepackage{amsmath,amsthm,amssymb}

\newcommand{\ket}[1]{|#1\rangle}

\newtheorem{thm}{Proposition}

\begin{document}

\title{Joining and splitting the quantum states of photons}

\author{Elsa Passaro}
\affiliation{ICFO-Institut de Ciencies Fotoniques, Castelldefels (Barcelona), Spain}
\author{Chiara Vitelli}
\affiliation{Dipartimento di Fisica, Sapienza Universit\`{a} di Roma, Roma, Italy}
\affiliation{Center of Life NanoScience @ La Sapienza, Istituto Italiano di Tecnologia, Roma, Italy}
\author{Nicol\`{o} Spagnolo}
\affiliation{Dipartimento di Fisica, Sapienza Universit\`{a} di Roma, Roma, Italy}
\author{Fabio Sciarrino}
\affiliation{Dipartimento di Fisica, Sapienza Universit\`{a} di Roma, Roma, Italy}
\author{Enrico Santamato}
\affiliation{Dipartimento di Fisica, Universit\`{a} di Napoli Federico II, Napoli, Italy}
\author{Lorenzo Marrucci}
\email{lorenzo.marrucci@unina.it}
\affiliation{Dipartimento di Fisica, Universit\`{a} di Napoli Federico II, Napoli, Italy}
\affiliation{CNR-SPIN, Complesso Universitario di Monte S. Angelo, Napoli, Italy}

\date{\today}

\begin{abstract}
A photonic process named as \emph{quantum state joining} has been recently experimentally demonstrated [C. Vitelli \textit{et al.}, Nature Photon. \textbf{7}, 521 (2013)] that corresponds to the transfer of the internal two-dimensional quantum states of two input photons, i.e., two photonic qubits, into the four-dimensional quantum state of a single photon, i.e., a photonic ququart. A scheme for the inverse process, namely \emph{quantum state splitting}, has also been theoretically proposed. Both processes can be iterated in a cascaded layout, to obtain the joining and/or splitting of more than two qubits, thus leading to a general scheme for varying the number of photons in the system while preserving its total quantum state, or quantum information content. Here, we revisit these processes from a theoretical point of view. After casting the theory of the joining and splitting processes in the more general photon occupation number notation, we introduce some modified schemes that are in principle unitary (not considering the implementation of the CNOT gates) and do not require projection and feed-forward steps. This can be particularly important in the quantum state splitting case, to obtain a scheme that does not rely on postselection. Moreover, we formally prove that the quantum joining of two photon states with linear optics requires the use of at least one ancilla photon. This is somewhat unexpected, given that the demonstrated joining scheme involves the sequential application of two CNOT quantum gates, for which a linear optical scheme with just two photons and postselection is known to exist. Finally we explore the relationship between the joining scheme and the generation of clusters of multi-particle entangled states involving more than one qubit per particle. 
\end{abstract}

\maketitle

\section{Introduction}
The emerging technology of quantum information promises a great enhancement of the computational power at our disposal, as well as the perfectly secure transmission of information \cite{bennett00,sergienko,ladd10}. To turn this vision into a reality, one of the greatest challenges today is to substantially increase the amount of information -- the number of qubits -- that can be processed simultaneously. In photonic approaches \cite{kok07,obrien09, pan12}, the number of qubits can be raised by increasing the number of photons. This is a fully scalable method, in principle, but in practice it is limited to 6-8 qubits by the present technology \cite{yao12}. An alternative approach is that of using an enlarged quantum dimensionality within the same photon, for example by combining different degrees of freedom, such as polarization, time-bin, wavelength, propagation paths, or transverse modes such as orbital angular momentum \cite{mair01,barreiro05,molina07,lanyon09,ceccarelli09,nagali10pra,straupe10,nagali10prl}. Although not scalable, the latter approach may allow for a substantial increase in the number of qubits \cite{gao10,pile11,dada11,munro12}. Ideally, one would therefore like to combine these two methods and be able to dynamically switch from one to the other, depending on the specific needs, even during the computational process itself.

To this purpose, a quantum process, called ``quantum state joining'', in which two arbitrary qubits initially encoded in separate input photons are combined into a single output photon, within a four-dimensional quantum space, has been recently introduced and experimentally demonstrated \cite{vitelli13}. The inverse process was also proposed, in which the four-dimensional quantum information carried in a single input photon is split into two output photons, each carrying a qubit \cite{vitelli13}. Both processes are in principle iterable \cite{vitelli13}, and hence may be used to realize an interface for converting a multi-photon encoding of quantum information into a single-photon higher-dimensional one and vice versa, thus enabling a full integration of the two encoding methods. These processes can be used to multiplex and demultiplex the quantum information in photons, for example with the purpose of using a smaller number of photons in lossy transmission channels. The idea of multiplexing/demultiplexing the quantum information in photons was for example proposed in \cite{garciaescartin08}, but a complete implementation scheme had not been developed (in particular, the proposal reported in Ref.\ \cite{garciaescartin08} relies on the existence of a hypothetical CNOT gate in polarization encoding between photons which is independent of the spatial mode of the photons, but the proposal does not include a discussion on how to realize this gate in practice) and had not been experimentally demonstrated. In addition, the quantum joining and splitting processes might also find application in the interfacing of multiple photonic qubits with a matter-based quantum register \cite{julsgaard04}, another crucial element of future quantum information networks \cite{kimble08}. For example, interfacing with multilevel quantum registers \cite{grace06} may be facilitated by the quantum joining/splitting schemes.

Let us first recall the main definitions of the quantum state joining and splitting processes. To make the language simpler and closer to the experimental demonstration reported in Ref.\ \cite{vitelli13}, we will initially refer to the polarization encoding of the qubits, although there is no general requirement on the choice of encoding at input and output. Let us then assume that two incoming photons, labeled 1 and 2, carry two polarization-encoded qubits, namely
\begin{align}
  \ket{\psi}_1 &= \alpha \ket{H}_{1} + \beta \ket{V}_{1}, \nonumber\\
  \ket{\phi}_2 &= \gamma \ket{H}_{2} + \delta \ket{V}_{2}, \label{inputphotons}
\end{align}
where $\ket{H}$ and $\ket{V}$ denote the states of horizontal and vertical linear polarization, corresponding to the logical 0 and 1, respectively. The two photons together form a (separable) quantum system, whose overall quantum state is given by the tensor product
\begin{align}
  \ket{\psi}_1 \otimes \ket{\phi}_2 &= \alpha \gamma \ket{H}_1 \ket{H}_2 +
  \alpha \delta \ket{H}_1 \ket{V}_2 \nonumber\\
  &+ \beta \gamma \ket{V}_1 \ket{H}_2 +
  \beta \delta \ket{V}_1 \ket{V}_2
\end{align}
The physical process of ``quantum state joining'' corresponds to transforming this two-photon system into a single-photon one, i.e., in an outcoming photon 3 having the following quantum state:
\begin{align}
  \ket{\Psi}_3 = \alpha \gamma \ket{0}_3 + \alpha \delta \ket{1}_3 + 
  \beta \gamma \ket{2}_3 + \beta \delta \ket{3}_3,
\end{align}
where $\ket{n}$ with $n=0,1,2,3$ are four arbitrary single-photon orthogonal states, defining a four-dimensional logical basis of a ququart. Of course we cannot use only the two-dimensional polarization encoding for the outgoing photon. One possibility is to use four independent spatial modes. Another option, adopted in Ref.\ \cite{vitelli13}, is to use two spatial modes combined with the two polarizations. In the latter case, in the words of Neergaard-Nielsen \cite{neergaard13}, ``the information is transferred from a Hilbert space of size 2 (photons) $\times$ 2 (polarizations) to a Hilbert space of size 1 (photon) $\times$ 2 (polarizations) $\times$ 2 (paths)''.

More generally, the joining process should work even for entangled qubits, both internally entangled (i.e., the two photons are entangled with each other) and externally entangled (the two photons are entangled with other particles outside the system). In the first case, the four coefficients obtained in the tensor product $\alpha\gamma, \alpha\delta, \beta\gamma, \beta\delta$ are replaced with four arbitrary coefficients $\alpha_0, \alpha_1, \alpha_2, \alpha_3$. In the second case, the four coefficients are replaced with four kets representing different quantum states of the external entangled system. Quantum splitting is defined as the inverse process, transforming an input photon 3 in the two photons 1 and 2.

The basic difficulty with implementing these state joining/splitting processes is that a form of interaction between photons is needed. But photons do not interact in vacuum and exhibit exceedingly weak interactions in ordinary nonlinear media. A way to introduce an effective interaction, known as the Knill-Laflamme-Milburn (KLM) method \cite{knill01}, is based on exploiting two-photon interferences and a subsequent ``wavefunction collapse'' occurring on measurement. This idea allowed for example the first experimental demonstrations of controlled-NOT (CNOT) quantum logical gates among qubits carried by different photons \cite{obrien03,pittman03,gasparoni04,zhao05}, and is now at the basis of the quantum joining process we are considering here.

To get the main idea of the quantum joining implementation, consider again the two input photons given in Eq.\ (\ref{inputphotons}). A single CNOT gate using one photon qubit as ``target'' and the other as ``control'' may be used to transfer a qubit from a photon to another, if the receiving photon initially carried a zeroed qubit. In order to obtain the state joining, we might then try for example to transfer the qubit $\phi$ from photon 2 to photon 1, while preserving the other qubit $\psi$ by storing it into a different degree of freedom of photon 1 (for example spatial modes). However, the interference processes utilized in the KLM CNOT require the two photons to be indistinguishable in everything, except for the qubit $\phi$ involved in the transfer. So, they are disrupted by the presence of the second qubit $\psi$ carried by the target photon, even if stored in different degrees of freedom.

To get around this obstacle, Vitelli \textit{et al.} in Ref.\ \cite{vitelli13} proposed a scheme that is based on the following three main subsequent steps: (i) ``unfold'' the target qubit $\psi$ (carried by input photon 1) initially travelling in mode $t$, by turning it into the superposition $\alpha\ket{H}_{t_1}+\beta\ket{H}_{t_2}$ of two zeroed polarization qubits, travelling in separate optical modes $t_1$ and $t_2$; (ii) duplicate the control qubit $\phi$ (carried by input photon 2) travelling in mode $c$ on an ancillary photon travelling in mode $a$, thus creating the entangled state $\gamma\ket{H}_c\ket{H}_a+\delta\ket{V}_c\ket{V}_a$; (iii) execute two KLM-like CNOT operations (of the Pittman kind \cite{pittman01,pittman03}), one with modes $c$ and $t_1$, the other with modes $a$ and $t_2$. In this way, each CNOT operates with a target photon that carries a zeroed qubit and no additional information, but the target photon is always interacting with either the control qubit or its entangled copy.

To complete the process, the photons travelling in modes $c$ and $a$ must be finally measured. For certain outcomes of this measurement, occurring with probability 1/32, the outgoing target photon is then collapsed in the final ``joined'' state $\ket{\psi\otimes\phi}=\alpha\gamma\ket{H}_{t_1} +\alpha\delta\ket{V}_{t_1} +\beta\gamma\ket{H}_{t_2} +\beta\delta\ket{V}_{t_2}$, which contains all the quantum information of the two input photons. The success probability can be raised to 1/8 by exploiting a feed-forward scheme and using other measurement outcomes. This probabilistic feature of the setup is common to all KLM-based implementations of CNOT gates (although, in principle, the success probability could be raised arbitrarily close to 100\% by using a large number of ancilla photons).

This paper is structured as follows. In Section II, the joining/splitting schemes are revisited by adopting the more general photon occupation-number formalism and some variants of the original schemes are introduced which do not need a projection and feed-forward mechanism to work (not considering the CNOT implementation), although at the price of using a doubled number of CNOT gates. In Section III, we then develop a formal proof of the fact that the quantum joining is impossible for an arbitrary linear optical scheme involving only two photons and a final post-selection step. Hence, at least one ancilla photon is needed (or the presence of optical nonlinearity). In Section IV, we analyze the relationship between the joining process of two photonic qubits and a particular class of three-photon entangled states, in which two photons are separately entangled with a common ``intermediate'' photon. We show that the quantum joining process can be used to create such cluster states and that, conversely, having at one's disposal one of these states, the quantum joining of two other photons can be immediately achieved by a teleportation scheme. We also note that these three-photon entangled states are of the same kind as the ``linked'' multiphoton states first introduced by Yoran and Reznik to perform deterministic quantum computation with linear optics \cite{yoran03}. Finally, in Section V, we draw some concluding remarks.

\section{Joining and splitting schemes in a photon-number notation}
In this Section, we revisit the joining/splitting schemes introduced in Ref.\ \cite{vitelli13} adopting the more general photon-number notation, as opposed to the polarization-ket notation used in the Introduction and in \cite{vitelli13}. In particular, photonic qubits will be represented as pairs of modes, with one photon that can occupy either one, as in the ``dual-rail'' qubit encoding. Of course, the two modes can also correspond to two orthogonal polarizations of a single spatial mode, thus reproducing the polarization-encoding case. 

Given two modes forming a qubit, the $\ket{10}$  ket, where the 0’s and 1’s refer here to the photon numbers, corresponds to having a photon in the first mode, encoding the logical 0 of the qubit. The $\ket{01}$  ket will then represent the photon in the second path, encoding the logical 1 of the qubit. For our schemes, we will however also need the $\ket{00}$ ket, representing a vacuum state, i.e. the ``empty'' qubit.
\begin{figure*}[htbp]
    \includegraphics[width=16cm]{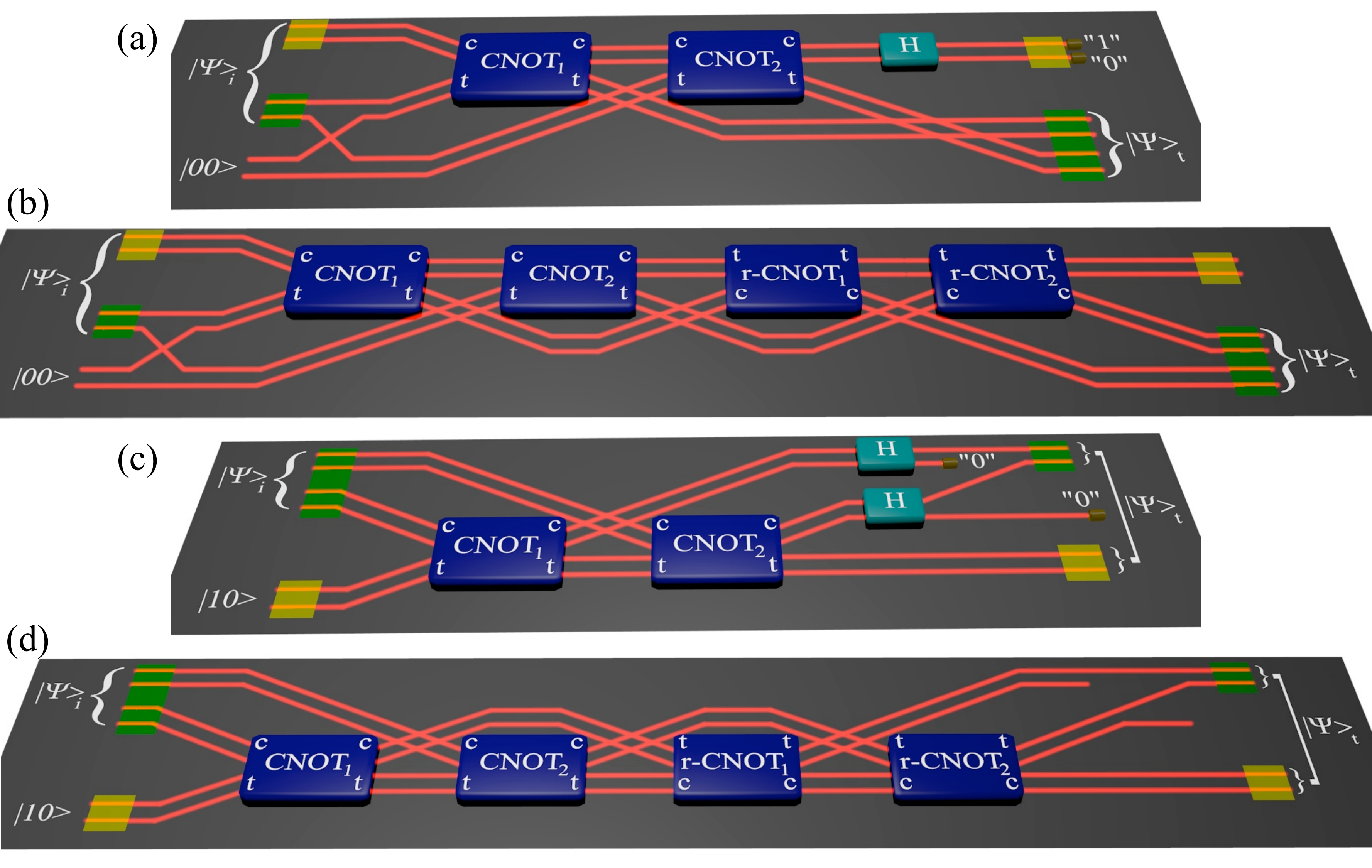}
    \caption{\label{figjoin} (Colors online) Optical schemes for quantum state joining and splitting. Each line represents a separate photonic mode (either spatial or of polarization). A qubit is represented by a double parallel line, as in a dual-rail implementation. A single-photon ququart, as obtained after the quantum joining, is represented by four parallel lines. H stands for a Hadamard quantum gate, CNOT for a controlled NOT quantum gate and r-CNOT for a CNOT gate in which the control and target ports have been reversed. The $``0''$ sign corresponds to a vacuum detection (no photons). $\ket{\Psi}_i$ is the input state and $\ket{\Psi}_t$ the final (target) state. (a) Scheme for quantum joining based on a double CNOT gate and final projection. Feed-forward is needed to obtain deterministic behavior (not considering the CNOT contribution). (b) Alternative scheme for deteministic quantum joining, using four CNOT gates, with the second two gates having inverted control and target ports. This leads to a deterministic behavior without projection and feed-forward (not considering the CNOT success probability). (c) Scheme for quantum splitting, with double CNOT gate and final projection. This scheme is probabilistic (with a 50\% success probability, not considering the CNOT gates contribution) and could be made deterministic only by combining quantum non-demolition measurements and feed-forward. (d) Alternative scheme for deterministic quantum splitting by using four CNOT gates (not considering the CNOT success probability).}
\end{figure*}

Let us first consider the joining process, schematically illustrated in Fig.\ \ref{figjoin}a. Labelling as $c$ (for control) and $t$ (for target) the travelling modes of the two input photons, the input state is taken to be the following:
\begin{eqnarray}
\ket{\Psi}_i &=& \alpha_0\ket{10}_t\ket{10}_c + \alpha_1\ket{10}_t\ket{01}_c \nonumber\\
&&+ \alpha_2\ket{01}_t\ket{10}_c + \alpha_3\ket{01}_t\ket{01}_c, \label{inputn}
\end{eqnarray}
which may represent both separable and non-separable two-photon states.

The qubit ``unfolding'' step corresponds to adding two empty modes for photon $t$ and rearranging the four modes so as to obtain the following state:
\begin{eqnarray}
\ket{\Psi}_u &=& \alpha_0\ket{1000}_t\ket{10}_c + \alpha_1\ket{1000}_t\ket{01}_c \nonumber\\
&& + \alpha_2\ket{0010}_t\ket{10}_c + \alpha_3\ket{0010}_t\ket{01}_c \nonumber\\
&=& \alpha_0\ket{10}_{t_1}\ket{00}_{t_2}\ket{10}_c + \alpha_1\ket{10}_{t_1}\ket{00}_{t_2}\ket{01}_c \nonumber\\
&&+ \alpha_2\ket{00}_{t_1}\ket{10}_{t_2}\ket{10}_c + \alpha_3\ket{00}_{t_1}\ket{10}_{t_2}\ket{01}_c,
\label{inputunf}
\end{eqnarray}
where in the second expression we have split the four $t$ modes, so as to treat the first two as one qubit ($t_1$) and the final two as a second qubit ($t_2$). Notice that both of them are initialized to logical zero, but with the possibility for each of them to be actually empty.

Each of these qubits must now be subject to a CNOT gate, using the same $c$ qubit as control. The action of the CNOT gate in the photon-number notation is described by the following equations:
\begin{eqnarray}
\hat{U}_{\text{CNOT}}\ket{10}_c\ket{10}_t&=&\ket{10}_c\ket{10}_t \nonumber\\
\hat{U}_{\text{CNOT}}\ket{01}_c\ket{10}_t&=&\ket{01}_c\ket{01}_t \nonumber\\
\hat{U}_{\text{CNOT}}\ket{10}_c\ket{01}_t&=&\ket{10}_c\ket{01}_t \\
\hat{U}_{\text{CNOT}}\ket{01}_c\ket{01}_t&=&\ket{01}_c\ket{10}_t \nonumber
\end{eqnarray}
However, in the present implementation of the quantum fusion we need to have the CNOT act also on ``empty target qubits'', that is vacuum states. For these we assume the following behavior:
\begin{eqnarray}
\hat{U}_{\text{CNOT}}\ket{10}_c\ket{00}_t&=&\eta\ket{10}_c\ket{00}_t \nonumber\\
\hat{U}_{\text{CNOT}}\ket{01}_c\ket{00}_t&=&\eta\ket{01}_c\ket{00}_t 
\end{eqnarray}
where $\eta$ is a possible complex amplitude rescaling relative to the non-vacuum case. A unitary CNOT must have $|\eta|=1$, but probabilistic implementations do not have this requirement. The quantum joining scheme works if the two CNOTs have the same $\eta$. In particular the CNOTs implementation proposed by Pittman et al.\ and used in \cite{vitelli13} have $\eta=1$, so for brevity we will remove $\eta$ in the following expressions.

Let us then consider the action of these two CNOT gates to the unfolded state given in Eq.\ (\ref{inputunf}):
\begin{eqnarray}
\ket{\Psi}_f&=&\hat{U}_{\text{CNOT}_2}\hat{U}_{\text{CNOT}_1}\ket{\Psi}_u \nonumber\\
&=& \alpha_0\ket{10}_{t1}\ket{00}_{t2}\ket{10}_c + \alpha_1\ket{01}_{t1}\ket{00}_{t2}\ket{01}_c \nonumber\\
&& + \alpha_2\ket{00}_{t1}\ket{10}_{t2}\ket{10}_c + \alpha_3\ket{00}_{t1}\ket{01}_{t2}\ket{01}_c \label{interm1}
\end{eqnarray}
If now we project the $c$ photon state on $\ket{+}=(\ket{10}+\ket{01})/\sqrt{2}$, so as to erase the $c$ qubit, and reunite the $t_1$ and $t_2$ kets, we obtain
\begin{equation}
\ket{\Psi}_{t}= \alpha_0\ket{1000}_t+\alpha_1\ket{0100}_t +\alpha_2\ket{0010}_t + \alpha_3\ket{0001}_t, \label{outputn}
\end{equation}
which is the desired joined state. Since the $c$ qubit measurement has a probability of 50\% of obtaining $\ket{+}$, without feed-forward the described method has a success probability of 50\% not considering the CNOT success probability.

If the outcome of the $c$ measurement is $\ket{-}=(\ket{10}-\ket{01})/\sqrt{2}$, we obtain the following target state:
\begin{equation}
\ket{\Psi'}_t=\alpha_0\ket{1000}_t-\alpha_1\ket{0100}_t+\alpha_2\ket{0010}_t-\alpha_3\ket{0001}_t.
\end{equation}
This state can be transformed back into $\ket{\Psi}_{t}$, as given in Eq.\ (\ref{outputn}), by a suitable unitary transformaton. Therefore, the success probability of the joining scheme can be raised to 100\% (again not considering CNOTs success probabilities) by a simple feed-forward mechanism.

Alternative to this feed-forward scheme, one might recover a deterministic behavior for the joining step (not considering the CNOT) by avoiding the $c$-photon projection and applying two additional CNOT gates in which control and target qubits have swapped roles, so as to ``disentangle'' the $c$ and $t$ photons. This alternative is illustrated in Fig.\ \ref{figjoin}b. In other words, after the first two CNOTs, we must apply a third CNOT with $t_1$ used as control and $c$ as target and a fourth CNOT with $t_2$ as control and $c$ as target. This time, for a proper working of the scheme, we must consider the possibility that the control port of the CNOT is empty. As for the previous case of empty target qubit, the CNOT outcome in this case is taken to be simply identical to the input except for a possible amplitude rescaling, i.e.
\begin{eqnarray}
\hat{U}_{\text{CNOT}}\vert00\rangle_c\vert10\rangle_t &=& \eta'\vert00\rangle_c\vert10\rangle_t \nonumber\\
\hat{U}_{\text{CNOT}}\vert00\rangle_c\vert01\rangle_t &=& \eta'\vert00\rangle_c\vert01\rangle_t
\end{eqnarray}
which is what occurs indeed in most CNOT implementations. Moreover, we again assume $\eta'=1$ in the following, for simplicity. Let us then take the state given in Eq.\ (\ref{interm1}) and apply the two ``reversed'' CNOT gates, denoted as r-CNOT:
\begin{eqnarray}
\ket{\Psi}_{f2}&=&\hat{U}_{\text{r-CNOT}_2}\hat{U}_{\text{r-CNOT}_1}\ket{\Psi}_f \nonumber\\
&=& \ket{10}_c \left( \alpha_0\ket{10}_{t1}\ket{00}_{t2}+\alpha_1\ket{01}_{t1}\ket{00}_{t2} \right. \nonumber\\
&& \left. +\alpha_2\ket{00}_{t1}\ket{10}_{t2} +\alpha_3\ket{00}_{t1}\ket{01}_{t2} \right) \\
&=& \ket{10}_c \otimes \ket{\Psi}_t \nonumber
\end{eqnarray}
where $\ket{\Psi}_t$ is given in Eq.\ (\ref{outputn}). After this step, one can just discard photon $c$ and photon $t$ will continue to hold the entire initial quantum information. We notice that this second implementation method does not require the feed-forward, which is an advantage in terms of resources, but it needs four CNOTs instead of two. Since CNOT implementations based on linear optics are actually probabilistic, the final success probability will be significantly smaller than the first method without feed-forward, so this scheme is not convenient at the present stage. 

Let us now move on to the quantum state splitting process, illustrated in Fig.\ \ref{figjoin}c.  We assume to have an input photon encoding two qubits (i.e., a ququart) in the four-path state
\begin{equation}
\ket{\psi}_i = \alpha_0\ket{1000}+\alpha_1\ket{0100}+\alpha_2\ket{0010} +\alpha_3\ket{0001}
\end{equation}
We label this input photon as $c$ (for control). We also label the first two modes as $c_1$ and the last two modes as $c_2$. We then take another photon, labeled as $t$ (for target), that is initialized in the logical zero state of two other modes, so that the initial two-photon state is the following:
\begin{eqnarray}
\ket{\Psi}_i &=& (\alpha_0\ket{10}_{c_1}\ket{00}_{c_2} + \alpha_1\ket{01}_{c_1}\ket{00}_{c_2} \nonumber\\
&&+ \alpha_2\ket{00}_{c_1}\ket{10}_{c_2} +\alpha_3\ket{00}_{c_1}\ket{01}_{c_2})\ket{10}_t \nonumber
\end{eqnarray}
We now apply the two CNOT gates in sequence, using the $t$ photon as target qubit in both cases and the $c_1$ and $c_2$ modes of the $c$ photon as control qubit in the first and second CNOT, respectively. In order to do these operations properly, we need to define the CNOT operation also for the case when the control qubit is empty, as already discussed above. Hence, we obtain
\begin{eqnarray}
&&\hat{U}_{\text{CNOT}_1}\hat{U}_{\text{CNOT}_2}\ket{\Psi}_i = \alpha_0\ket{10}_{c_1}\ket{00}_{c_2}\ket{10}_t \nonumber\\
&&+ \alpha_1\ket{01}_{c_1}\ket{00}_{c_2}\ket{01}_t + \alpha_2\ket{00}_{c_1}\ket{10}_{c_2}\ket{10}_t \nonumber\\
&&+ \alpha_3\ket{00}_{c_1}\ket{01}_{c_2}\ket{01}_t\nonumber
\end{eqnarray}
Now we need to erase part of the information contained in the control photon. This is accomplished by projecting onto $\ket{+}_{c_i}=(\ket{10}_{c_i}+\ket{01}_{c_i})/\sqrt{2}$ combinations of the first and second pairs of modes, while keeping unaffected their relative amplitudes. In other words, we must apply an Hadamard transformation on each pair of modes, and take as successful outcome only the logical-zero output (corresponding to the $\ket{+}$ combination of the inputs). The projection is actually performed by checking that no photon comes out of the $\ket{-}$ (i.e., logical one) output ports of the Hadamard. The two surviving output modes are then combined into a single output $c$-photon qubit, which together with the $t$-photon qubit form the desired split-qubit output. Indeed, we obtain the following projected output:
\begin{eqnarray}
\ket{\Psi}_f &=& \alpha_0\ket{10}_{c}\ket{10}_t + \alpha_1\ket{10}_{c}\ket{01}_t \nonumber\\
&& + \alpha_2\ket{01}_{c}\ket{10}_t +\alpha_3\ket{01}_{c}\ket{01}_t \label{outputsplit}
\end{eqnarray}
which describes the same two-qubit state as the input, but encoded in two photons instead of one. The proposed scheme for splitting has a 50\% probability of success, not considering the CNOT contribution. It might be again possible to bring the probability to 100\% (not considering CNOTs) by detecting the actual $c$-photon output mode pair after the Hadamard gates by a quantum non-demolition approach or in post-selection, and then applying an appropriate unitary transformation to the $t$ photon.

Also in this case, we can replace the projection and feed-forward scheme by the action of a third and fourth CNOT gates in which the target and control roles are reversed, that is, using the $t$ photon as control and $c_1$ and $c_2$ as targets of the third and fourth CNOT gates, respectively. This is shown in Fig.\ \ref{figjoin}d. After the four CNOTs, one obtains the following state:
\begin{eqnarray}
&&\alpha_0\ket{10}_{c_1}\ket{00}_{c_2}\ket{10}_t + \alpha_1\ket{10}_{c_1}\ket{00}_{c_2}\ket{01}_t \nonumber\\
&&+ \alpha_2\ket{00}_{c_1}\ket{10}_{c_2}\ket{10}_t + \alpha_3\ket{00}_{c_1}\ket{10}_{c_2}\ket{01}_t \nonumber\\
&& = \alpha_0\ket{1000}_c\ket{10}_t + \alpha_1\ket{1000}_c\ket{01}_t \nonumber\\
&&+ \alpha_2\ket{0010}_c\ket{10}_t + \alpha_3\ket{0010}_c\ket{01}_t,
\end{eqnarray}
where in the second equality we have regrouped the four $c$ modes. Then, an inverse unfolding step, that is simply discarding the second and fourth mode of the $c$ photon, which are always empty, will lead to final state $\ket{\Psi}_f$ given in Eq.\ (\ref{outputsplit}) with 100\% probability. In this splitting case, the advantage of using this alternative scheme is more marked, as it is the only possibility to avoid postselection or quantum nondemolition steps.

The CNOT gates utilized in the joining and splitting processes described in this Section can be implemented using different methods. In particular, since the photons being processed in each CNOT stage have no additional information, linear-optics KLM-like schemes based on two-photon interference can be used. It is for this reason that our schemes require the unfolding step and a double CNOT, rather than using a single CNOT for transferring the qubit from one photon to the other (if nonlinear-optical CNOT gates will ever be realized, they might possibly allow for a CNOT operation to be performed while another degree of freedom is present and remains unaffected, thus making the joining/splitting schemes much simpler). The only requirement for these CNOT gates is that they must be applicable also to the case when one of the input qubits is empty, i.e., there is a vacuum state at one input port. We shall see in the next Section, that this is a nontrivial requirement. 

\section{Unfeasibility of the quantum joining with two photons and post-selection}
\label{sect:bs}
There exist different linear-optical based implementations of CNOT gates. The simplest are those based on post-selection and not requiring ancillary photons, such as the scheme first proposed by Ralph et al.\ \cite{ralph02} and Hofmann and Takeuchi \cite{hofman02} and later experimentally demonstrated by O'Brien et al.\ \cite{obrien03}. Although such CNOT gates are based on post-selection and hence require destroying the output photons, there exist also schemes for applying several CNOT gates in sequence, with only one final post-selection step \cite{ralph04}. These schemes require only the two photons to be combined and a final post-selection step based on photon detection. Given the CNOT-based general scheme for quantum joining described in the previous Section, it is then natural to try an implementation exploiting these schemes.

More generally, one might ask whether a proper mixing of the two photon modes in a suitable linear-optical setup, followed by a filtering step on one of the two photons might suffice to obtain the joining onto the remaining photon. In this Section we prove that this is not possible. No possible unitary evolution of the two photons as resulting from propagation through an arbitrary linear optical system, followed by an arbitrary projection for one of the two photons can lead to the quantum joining. This in turn shows that the joining scheme cannot be based on the CNOT gates of Ralph's kind and requires at least one ancilla photon. With one ancilla photon, it is possible to implement for example the CNOT gates proposed by Pittman \cite{pittman01} and thus successfully obtain the quantum joining of two photon states, as demonstrated in Ref.\ \cite{vitelli13}. Of course the demonstrated implementation is probabilistic, because the CNOT gates are implemented in a probabilistic way.

The statement we prove is the following:
\begin{thm}
It is impossible to transfer all the quantum information encoded in two input photons into one output photon using linear optics and post-selection, without including ancillary photons in the process.
\end{thm}

\begin{proof}
The two-photon input state can be written in full generality as follows:
\begin{eqnarray}
\ket{\Psi}_i &=& \alpha_0 \ket{1010} + \alpha_1\ket{1001} + \alpha_2\ket{0110} + \alpha_3\ket{0101} \nonumber\\
&=& \left( \alpha_0 \hat{a}^+_1 \hat{a}^+_3 + \alpha_1 \hat{a}^+_1 \hat{a}^+_4 + \alpha_2 \hat{a}^+_2 \hat{a}^+_3 + \alpha_3 \hat{a}^+_2 \hat{a}^+_4 \right) \ket{\emptyset} \nonumber\\
&& \label{inputstate}
\end{eqnarray}
where the four coefficients $\alpha_0, \alpha_1, \alpha_2, \alpha_3$ define the input quantum information, $\hat{a}^+_i$ denote the creation operators for an arbitrary orthonormal set of input modes $\ket{\psi}_i$, and $\ket{\emptyset}$ denotes the global vacuum state. The mode-indices $i$ here can be taken to include also the polarization degree of freedom, and we have selected four arbitrary modes 1-4 to encode the input information, with modes 1-2 used for one qubit and modes 3-4 for the other (possibly entangled with each other).

The propagation through an arbitrary linear-optical system can be described by the following transformation of the creation operators:
\begin{equation}
\hat{a}^{+}_i \rightarrow \hat{b}^{+}_i = \sum_{j} u_{ij} \hat{a}^{+}_j
\label{optprop}
\end{equation}
where $u_{ij}$ are the coefficients of a unitary matrix describing the propagation and the operators $\hat{b}^+_j$ create the propagated (output) modes $\ket{\chi}_j$. Applying this transformation to the input state Eq.\ (\ref{inputstate}) we obtain the following propagated two-photon state (here we are using the Schroedinger representation, in which the evolution acts on the state):
\begin{equation}
\ket{\Psi}_p  = \left( \alpha_0 \hat{b}^+_1 \hat{b}^+_3 + \alpha_1 \hat{b}^+_1 \hat{b}^+_4 + \alpha_2 \hat{b}^+_2 \hat{b}^+_3 + \alpha_3 \hat{b}^+_2 \hat{b}^+_4 \right) \ket{\emptyset}
\label{propagatedstate}
\end{equation}
Now, let us act on this state with a projector $\hat{\Pi}$ corresponding to the detection of a single photon in the arbitrary mode $\ket{\phi}=\sum_h \phi_h \ket{\chi}_h$, as given by the following:
\begin{equation}
\hat{\Pi} = \sum_h \phi_h^* \hat{b}_h .
\end{equation}
Thus, we obtain the following projected one-photon state
\begin{eqnarray}
\ket{\Psi}_f  &=& \hat{\Pi}\ket{\Psi}_p = \sum_h \phi_h^* \hat{b}_h \left( \alpha_0 \hat{b}^+_1 \hat{b}^+_3 + \alpha_1 \hat{b}^+_1 \hat{b}^+_4 \right. \nonumber\\
&& \left. + \alpha_2 \hat{b}^+_2 \hat{b}^+_3 + \alpha_3 \hat{b}^+_2 \hat{b}^+_4 \right) \ket{\emptyset} \label{finalstate}\\
&=& \alpha_0 \left( \phi_3^* \ket{\chi}_1 + \phi_1^* \ket{\chi}_3\right) +
\alpha_1 \left( \phi_4^* \ket{\chi}_1 + \phi_1^* \ket{\chi}_4\right) \nonumber\\
&& + \alpha_2 \left( \phi_3^* \ket{\chi}_2 + \phi_2^* \ket{\chi}_3\right) +
\alpha_3 \left( \phi_4^* \ket{\chi}_2 + \phi_2^* \ket{\chi}_4\right) \nonumber
\end{eqnarray}
Hence, owing to the bosonic nature of the photons, the final state results to be a linear combination of the following four ``symmetrized'' optical modes
\begin{eqnarray}
\ket{u}_0 &=& \phi_3^* \ket{\chi}_1 + \phi_1^* \ket{\chi}_3, \nonumber\\
\ket{u}_1 &=& \phi_4^* \ket{\chi}_1 + \phi_1^* \ket{\chi}_4, \nonumber\\
\ket{u}_2 &=& \phi_3^* \ket{\chi}_2 + \phi_2^* \ket{\chi}_3, \\
\ket{u}_3 &=& \phi_4^* \ket{\chi}_2 + \phi_2^* \ket{\chi}_4. \nonumber
\end{eqnarray}
The input quantum information will be preserved if and only if the four modes $\ket{u}_i$ form a linearly independent set. This in turn will depend on the determinant of the following matrix $\mathcal{M}$ of coefficients, expressing the linear dependence of the four $\ket{u}_i$ modes on the propagated modes $\ket{\chi}_i$:
\begin{align}
    \mathcal{M} = \left( 
    \begin{array}{cccc}
      \phi_3^* & 0 & \phi_1^* & 0  \\
      \phi_4^* & 0 & 0 & \phi_1^*  \\
      0 & \phi_3^* & \phi_2^* & 0  \\
      0 & \phi_4^* & 0 & \phi_2^* 
    \end{array} \right),
\end{align}
A simple calculation shows that the determinant of this matrix is identically nil, thus proving the statement. If the optical system includes losses from media absorption, these can be included in the treatment as additional non-optical excitation modes in which the input optical modes can be transformed in the course of propagation. In other words, Eq.\ (\ref{optprop}) will include also the creation operators of material excitations, although the latter will not contribute to the $\ket{u}_i$ and $\ket{\phi}$ modes. Thus the proof remains valid even in lossy optical systems.
\end{proof}

As a consequence of our proof, we can state that in general the mixing of two photons in a linear optical scheme followed by a final postselection step can only lead to a loss of some information, e.g., ending up with a qutrit instead of a ququart. Alternatively, one may somehow preserve the initial information conditioned on the fact that there is ``less information to start with'', because the input two-photon state is properly constrained, for example, to a separable state \cite{grudka03,bartuskova06}.

Beside this mathematical proof, one might be interested in seeking a more physical explanation for why the joining scheme using two CNOT in series following the concept of Ref.\ \cite{ralph04} fails. To this purpose, we carried out a detailed analysis, of which we report here only the main conclusions. The problem is that the scheme given in Ref.\ \cite{ralph04} is conceived for executing multiple CNOT in series, with the assumption that each control or target port of all the gates is occupied by a photon carrying the corresponding qubit. In the case of quantum joining, instead, the target ports may see the presence of ``empty'' qubits (or vacuum states), which open up new possible photonic evolution channels in the setup that are not excluded in the final post-selection step and which are instead absent in the standard case. These channels alter the final probabilities and disrupt the CNOT proper workings.

\section{Three-photon entangled states}
\label{sect:3ent}
In this Section, we explore the relationship between the joining process of two photonic qubits and a particular class of three-photon entangled states (TPES), in which two photons are separately entangled with a common ``intermediate'' photon. This intermediate doubly-entangled photon must clearly hold two separate qubits, as defined by exploiting four orthogonal optical modes. A schematic diagram of this particular form of entangled cluster is given in Fig.\ \ref{figcluster}.
\begin{figure}[htbp]
    \includegraphics[width=8.4cm]{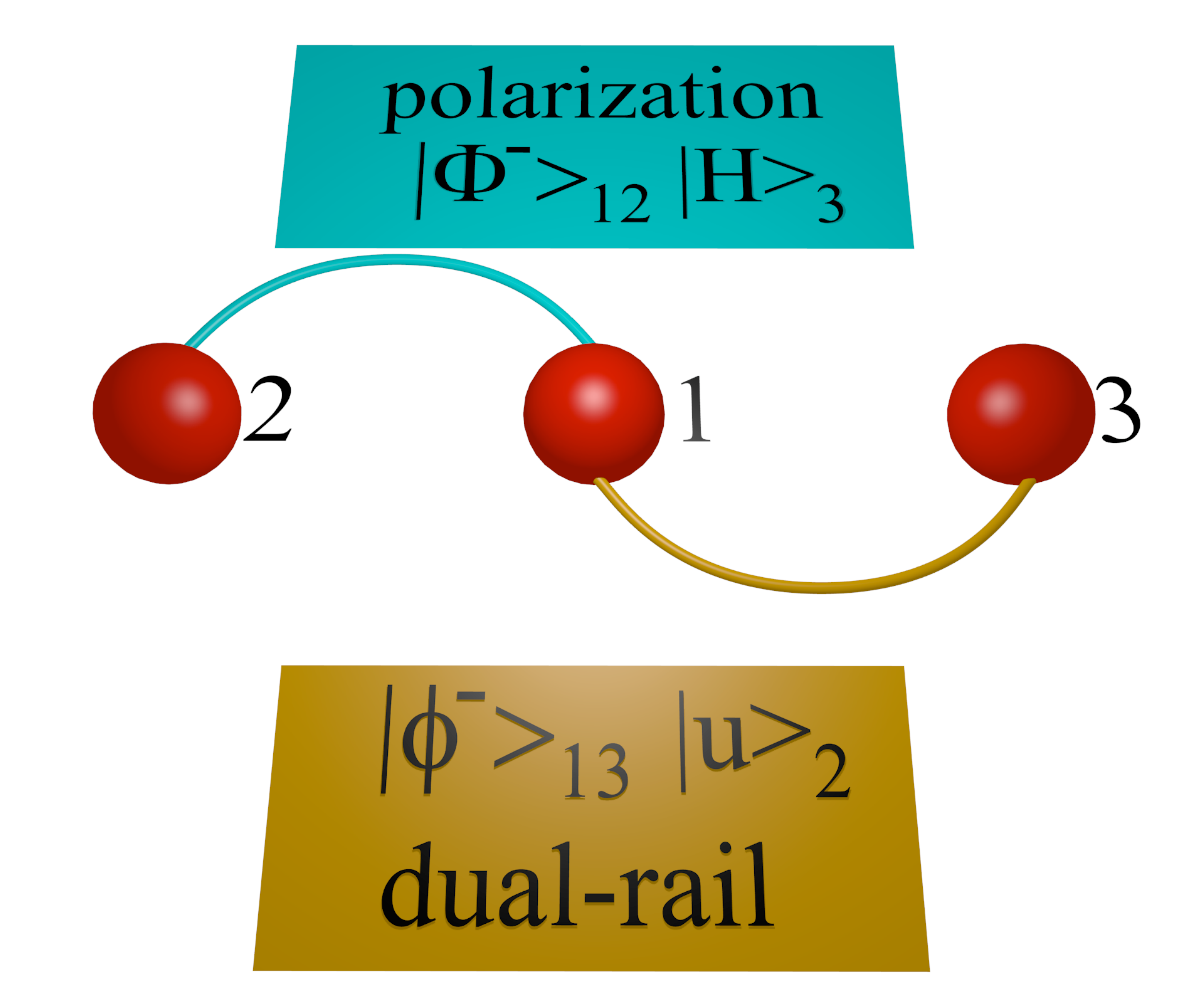}
    \caption{\label{figcluster} (Colors online) Schematic representation of the TPES state given in Eq.\ (\ref{eq:entphoton}) of the main text. The three red balls represent the three photons, with photon 1 separately entangled with photons 2 and 3. Each line correspond to an entanglement link. In the specific example, the upper (green) line corresponds to a polarization maximal entanglement in Bell state $\ket{\Phi^{-}}$, while the lower (yellow) line indicates a spatial-mode maximal entanglement (in a dual-rail basis) in Bell state $\ket{\phi^{-}}$. Other TPES states can be obtained by changing the specific Bell states used for the two entanglement links.}
\end{figure}

An example of such three-photon entangled states can be defined as follows:
\begin{align}
  \ket{\psi}_{123} &= \frac{1}{2} \left( \ket{H}_1\ket{V}_2-\ket{V}_1\ket{H}_2 \right) \ket{H}_3 \nonumber \\
  &\otimes \, \left(\ket{u}_1\ket{d}_3-\ket{d}_1\ket{u}_3 \right)\ket{u}_2, \label{eq:entphoton}
\end{align}
where photon 1 is the intermediate photon, entangled with photons 2 and 3, and we introduced $\ket{u}$ and $\ket{d}$, to refer to the ``up'' and ``down'' paths of a dual-rail qubit encoding. It is understood that the three photons are identified by a further label associated with propagation modes. Other possible pair-wise maximally-entangled TPES are obtained from \eqref{eq:entphoton} by exchanging $H$ and $V$ and/or $u$ and $d$ in photon 1, or by changing the $-$ sign into a $+$ in one or both the factors, for a total of 16 possible independent combinations. These can be also written more compactly in terms of the Bell basis of maximally entangled qubit pairs, defined as follows
\begin{subequations}\label{eq:bellbasis}
  \begin{align}\label{eq:bellbasis1}
    \ket{\Psi^{\pm}}_{ij} &= \frac{1}{\sqrt{2}}(\ket{H}_i \ket{H}_j \pm 
    \ket{V}_i \ket{V}_j) \ket{u}_i \ket{u}_j, \\ 
    \ket{\Phi^{\pm}}_{ij} &= \frac{1}{\sqrt{2}}(\ket{H}_i \ket{V}_j \pm 
    \ket{V}_i \ket{H}_j) \ket{u}_i \ket{u}_j, \label{eq:bellbasis2} \\
    \ket{\psi^{\pm}}_{ij} &= \frac{1}{\sqrt{2}}(\ket{u}_i \ket{u}_j \pm 
    \ket{d}_i \ket{d}_j) \ket{H}_i \ket{H}_j, \label{eq:bellbasis3} \\
    \ket{\phi^{\pm}}_{ij} &= \frac{1}{\sqrt{2}}(\ket{u}_i \ket{d}_j \pm 
    \ket{d}_i \ket{u}_j) \ket{H}_i \ket{H}_j. \label{eq:bellbasis4}
  \end{align}
\end{subequations}
Using this notation, Eq.\ \eqref{eq:entphoton}  can be for example rewritten as $\ket{\psi}_{123}=\ket{\Phi^{-}}_{12}\ket{H}_3\otimes\ket{\phi^{-}}_{13}\ket{u}_2$, and the other TPES are obtained by replacing one or both the Bell states with another one. Notice that states \eqref{eq:bellbasis1} and \eqref{eq:bellbasis2} are Bell states with respect to the polarization degree of freedom of the pair, while states \eqref{eq:bellbasis3} and \eqref{eq:bellbasis4} are Bell states with respect to the spatial-mode (or propagation path) degree of freedom of the pair.

Here and in the following discussion, for the sake of definiteness, we have adopted a notation referring to the specific case in which the double entanglement exploits two separate degrees of freedom, that is the polarization and a pair of spatial modes. We stress, however, that there is no actual requirement of using this specific choice of degrees of freedom for the validity of our analysis. There is even no need of using two separate degrees of freedom, as all qubit entanglements could for example also be encoded using a set of four spatial modes, e.g., four parallel paths, or four eigenmodes of the orbital angular momentum of light.

The state \eqref{eq:entphoton}, or anyone of the other TPES, can be used as a resource for carrying out a two-in-one qubit teleportation of the quantum state, i.e., the teleportation of the four-dimensional quantum state initially encoded in two input photons in a single output photon. Actually, the problem of preparing three-photon entangled states such as \eqref{eq:entphoton} is essentially equivalent to that of realizing the quantum joining. Indeed, quantum state joining can be used to prepare the three-photon entangled state \eqref{eq:entphoton} and, conversely, state \eqref{eq:entphoton} can be used to carry out the quantum state joining of two photonic qubits via teleportation. Let us see this in more detail.

To obtain the three-photon state \eqref{eq:entphoton}, or anyone of the other TPES, one must simply apply the quantum state joining protocol to two photons each taken from a separate entangled pair. In particular, one pair (say, photons 2 and 4) must be entangled in polarization and the other (photons 3 and 5) in the spatial degree of freedom defined by modes $\ket{u}$ and $\ket{d}$. Then photons 4 and 5 are state-joined into photon 1, so that their polarization and spatial modes properties are both transferred into this single photon. This leads immediately to state \eqref{eq:entphoton}.

Conversely, let us assume that we have initially three photons (labeled 1, 2, 3) in state \eqref{eq:entphoton} and that the qubits we want to join are encoded in two other photons (labeled 4 and 5), as described by the states 
\begin{align}
  \ket{\psi}_4 &= (\alpha \ket{H}_4 + \beta \ket{V}_4) \otimes \ket{u}_4,\\
  \ket{\psi}_5 &= \ket{H}_5 \otimes (\gamma \ket{u}_5 + \delta \ket{d}_5).
\end{align}

Recasting the overall 5-photon initial state $\ket{\psi}_{12345}=\ket{\psi}_{123} \ket{\psi}_{4} \ket{\psi}_{5} $ in terms of the 
basis \eqref{eq:bellbasis} for the states of the pairs 2, 4 and 3, 5, one obtains the following expression:
\begin{widetext}
  \begin{align}
    \ket{\psi}_{12345} &= \frac{1}{4} \ket{\Phi^+}_{24} 
    \ket{\phi^+}_{35} (\alpha \ket{H}_1 - \beta \ket{V}_1) 
    (\gamma \ket{u}_1 - \delta \ket{d}_1) 
    - \frac{1}{4} \ket{\Phi^+}_{24} \ket{\phi^-}_{35}(\alpha \ket{H}_1 
    - \beta \ket{V}_1) (\gamma \ket{u}_1 + \delta \ket{d}_1) 
    \nonumber \\
    &+ \frac{1}{4} \ket{\Phi^+}_{24} \ket{\psi^+}_{35}(\alpha \ket{H}_1 
    - \beta \ket{V}_1) 
    (\delta \ket{u}_1 - \gamma \ket{d}_1)
    - \frac{1}{4} \ket{\Phi^+}_{24} \ket{\psi^-}_{35}(\alpha \ket{H}_1 
    - \beta \ket{V}_1) 
    (\delta \ket{u}_1 + \gamma \ket{d}_1) \nonumber \\
    &- \frac{1}{4} \ket{\Phi^-}_{24} \ket{\phi^+}_{35}(\alpha \ket{H}_1 
    +\beta \ket{V}_1) 
    (\gamma \ket{u}_1 - \delta \ket{d}_1) 
    + \frac{1}{4} \ket{\Phi^-}_{24} \ket{\phi^-}_{35}(\alpha \ket{H}_1 
    + \beta \ket{V}_1) 
    (\gamma \ket{u}_1 + \delta \ket{d}_1) \nonumber \\
    &- \frac{1}{4} \ket{\Phi^-}_{24} \ket{\psi^+}_{35}(\alpha \ket{H}_1 
    + \beta \ket{V}_1) 
    (\delta \ket{u}_1 - \gamma \ket{d}_1) 
    + \frac{1}{4} \ket{\Phi^-}_{24} \ket{\psi^-}_{35}(\alpha \ket{H}_1 
    + \beta \ket{V}_1) 
    (\delta \ket{u}_1 + \gamma \ket{d}_1) \nonumber \\
    &+ \frac{1}{4} \ket{\Psi^+}_{24} \ket{\phi^+}_{35}(\beta \ket{H}_1 
    - \alpha \ket{V}_1) 
    (\gamma \ket{u}_1 - \delta \ket{d}_1) 
    - \frac{1}{4} \ket{\Psi^+}_{24} \ket{\phi^-}_{35}(\beta \ket{H}_1 
    - \alpha \ket{V}_1) 
    (\gamma \ket{u}_1 + \delta \ket{d}_1) \nonumber \\
    &+ \frac{1}{4} \ket{\Psi^+}_{24} \ket{\psi^+}_{35}(\beta \ket{H}_1 
    - \alpha \ket{V}_1)
    (\delta \ket{u}_1 - \gamma \ket{d}_1) 
    - \frac{1}{4} \ket{\Psi^+}_{24} \ket{\psi^-}_{35}(\beta \ket{H}_1 
    - \alpha \ket{V}_1) 
    (\delta \ket{u}_1 + \gamma \ket{d}_1) \nonumber \\
    &- \frac{1}{4} \ket{\Psi^-}_{24} \ket{\phi^+}_{35}(\beta \ket{H}_1 
    + \alpha \ket{V}_1) 
    (\gamma \ket{u}_1 - \delta \ket{d}_1) 
    + \frac{1}{4} \ket{\Psi^-}_{24} \ket{\phi^-}_{35}(\beta \ket{H}_1 
    + \alpha \ket{V}_1) 
    (\gamma \ket{u}_1 + \delta \ket{d}_1) \nonumber \\
    &- \frac{1}{4} \ket{\Psi^-}_{24} \ket{\psi^+}_{35}(\beta \ket{H}_1 
    + \alpha \ket{V}_1) 
    (\delta \ket{u}_1 - \gamma \ket{d}_1) 
    + \frac{1}{4} \ket{\Psi^-}_{24} \ket{\psi^-}_{35}(\beta \ket{H}_1 
    + \alpha \ket{V}_1) 
    (\delta \ket{u}_1 + \gamma \ket{d}_1).
  \end{align}
\end{widetext}

Then, to obtain the state-joining one needs to perform a Bell measurement in polarization on the photons 2 and 4 and another Bell measurement in the modes $u$ and $d$ on the photons 3 and 5. Whatever the outcome of the two Bell measurements, one may carry out an appropriate unitary transformation in order to cast photon 1 in the desired ``joined'' state 
$$\ket{\Psi}_1 = (\alpha \ket{H}_1 + \beta \ket{V}_1) \otimes (\delta \ket{u}_1 + \gamma \ket{d}_1).$$ The unitary transformation must be selected according to the result of the two Bell measurements, out of 16 possible results (and if a different TPES state is used in the process, it affects only the set of unitary transformations to be used).

It is interesting to note that, if a method for deterministic complete Bell state measurement is available, the quantum state joining obtained by this teleportation method can be also accomplished in a deterministic way. Indeed, one needs to prepare in advance a TPES using the probabilistic joining protocol by making as many attempts as needed. Then, one can complete the joining of the input photons deterministically by using the above described teleportation protocol.
  
Following the same ideas, one may use the state joining protocol to create even more complex entanglement clusters of photons, exploiting multiple degrees of freedom per photon. In particular, we notice that the TPES introduced above belong to the family of ``linked'' states first proposed by Yoran and Reznik in order to perform deterministic quantum computation with linear optics \cite{yoran03}. Not surprisingly, the optical scheme proposed in Ref.\ \cite{yoran03} to create such linked states is also very similar to that used for quantum-state joining (but it did not include explicitly the KLM gate implementations).

\section{Conclusion}
In summary, we have revisited the quantum-state joining and splitting processes recently introduced in Ref.\ \cite{vitelli13}. After casting the associated formalism in the more general photon-number notation, we have introduced some modified schemes that do not require feed-forward or post-selection. Next, we have provided a formal proof that the quantum joining of two photon states with linear optics cannot be accomplished using only post-selection and requires the use of at least one ancilla photon, despite the existence of linear-optical implementations of CNOT gates which do not require ancillary photons. Finally we have investigated the relationship between the state-joining scheme and the generation of clusters of three-photon entangled states involving more than one qubit per particle.

These schemes for multiplexing the quantum information across photons, despite having a relatively low success probability, may already find application in quantum communication or in interfacing with atomic memories, when high losses are involved. The generation of complex clusters of entangled photons may be of fundamental interest and might open the way to novel quantum protocols. Finally, we notice that if the recent attempts at achieving gigantic nonlinear interactions among photons will succeed \cite{shahmoon11,peyronel12}, deterministic schemes for quantum-state joining and splitting should also become possible, likely making the associated photon multiplexing/demultiplexing an important resource for future quantum communication networks. 

\section*{Acknowledgments}
This work was supported by the 7$^{th}$ Framework Programme of the European Commission, within the Future Emerging Technologies program, under grant No.\ 255914, PHORBITECH. E.P. acknowledges financial support by the Spanish project FIS2010-14830 and Generalitat de Catalunya, and by the European PERCENT ERC Starting Grant. 

\bibliographystyle{apsrev}

\begin{thebibliography}{39}
\expandafter\ifx\csname natexlab\endcsname\relax\def\natexlab#1{#1}\fi
\expandafter\ifx\csname bibnamefont\endcsname\relax
  \def\bibnamefont#1{#1}\fi
\expandafter\ifx\csname bibfnamefont\endcsname\relax
  \def\bibfnamefont#1{#1}\fi
\expandafter\ifx\csname citenamefont\endcsname\relax
  \def\citenamefont#1{#1}\fi
\expandafter\ifx\csname url\endcsname\relax
  \def\url#1{\texttt{#1}}\fi
\expandafter\ifx\csname urlprefix\endcsname\relax\def\urlprefix{URL }\fi
\providecommand{\bibinfo}[2]{#2}
\providecommand{\eprint}[2][]{\url{#2}}

\bibitem[{\citenamefont{Bennett and DiVincenzo}(2000)}]{bennett00}
\bibinfo{author}{\bibfnamefont{C.~H.} \bibnamefont{Bennett}} \bibnamefont{and}
  \bibinfo{author}{\bibfnamefont{D.~P.} \bibnamefont{DiVincenzo}},
  \bibinfo{journal}{Nature} \textbf{\bibinfo{volume}{404}},
  \bibinfo{pages}{247} (\bibinfo{year}{2000}).

\bibitem[{\citenamefont{Sergienko}(2006)}]{sergienko}
\bibinfo{author}{\bibfnamefont{A.~V.} \bibnamefont{Sergienko}},
  \emph{\bibinfo{title}{Quantum communications and cryptography}}
  (\bibinfo{publisher}{Taylor \& Francis Group}, \bibinfo{year}{2006}).

\bibitem[{\citenamefont{Ladd et~al.}(2010)\citenamefont{Ladd, Jelezko,
  Laflamme, Nakamura, Monroe, and O'Brien}}]{ladd10}
\bibinfo{author}{\bibfnamefont{T.~D.} \bibnamefont{Ladd}},
  \bibinfo{author}{\bibfnamefont{F.}~\bibnamefont{Jelezko}},
  \bibinfo{author}{\bibfnamefont{R.}~\bibnamefont{Laflamme}},
  \bibinfo{author}{\bibfnamefont{Y.}~\bibnamefont{Nakamura}},
  \bibinfo{author}{\bibfnamefont{C.}~\bibnamefont{Monroe}}, \bibnamefont{and}
  \bibinfo{author}{\bibfnamefont{J.~L.} \bibnamefont{O'Brien}},
  \bibinfo{journal}{Nature} \textbf{\bibinfo{volume}{464}}, \bibinfo{pages}{45}
  (\bibinfo{year}{2010}).

\bibitem[{\citenamefont{Kok et~al.}(2007)\citenamefont{Kok, Munro, Nemoto,
  Ralph, Dowling, and Milburn}}]{kok07}
\bibinfo{author}{\bibfnamefont{P.}~\bibnamefont{Kok}},
  \bibinfo{author}{\bibfnamefont{W.~J.} \bibnamefont{Munro}},
  \bibinfo{author}{\bibfnamefont{K.}~\bibnamefont{Nemoto}},
  \bibinfo{author}{\bibfnamefont{T.~C.} \bibnamefont{Ralph}},
  \bibinfo{author}{\bibfnamefont{J.~P.} \bibnamefont{Dowling}},
  \bibnamefont{and} \bibinfo{author}{\bibfnamefont{G.~J.}
  \bibnamefont{Milburn}}, \bibinfo{journal}{Rev.\ Mod.\ Phys.}
  \textbf{\bibinfo{volume}{79}}, \bibinfo{pages}{135} (\bibinfo{year}{2007}).

\bibitem[{\citenamefont{O'Brien et~al.}(2009)\citenamefont{O'Brien, Furusawa,
  and Vuckovi\'c}}]{obrien09}
\bibinfo{author}{\bibfnamefont{J.~L.} \bibnamefont{O'Brien}},
  \bibinfo{author}{\bibfnamefont{A.}~\bibnamefont{Furusawa}}, \bibnamefont{and}
  \bibinfo{author}{\bibfnamefont{J.}~\bibnamefont{Vuckovi\'c}},
  \bibinfo{journal}{Nat.\ Photon.} \textbf{\bibinfo{volume}{3}},
  \bibinfo{pages}{687} (\bibinfo{year}{2009}).

\bibitem[{\citenamefont{Pan et~al.}(2012)\citenamefont{Pan, Chen, Lu,
  Weinfurter, Zeilinger, and Zukowski}}]{pan12}
\bibinfo{author}{\bibfnamefont{J.-W.} \bibnamefont{Pan}},
  \bibinfo{author}{\bibfnamefont{Z.-B.} \bibnamefont{Chen}},
  \bibinfo{author}{\bibfnamefont{C.-Y.} \bibnamefont{Lu}},
  \bibinfo{author}{\bibfnamefont{H.}~\bibnamefont{Weinfurter}},
  \bibinfo{author}{\bibfnamefont{A.}~\bibnamefont{Zeilinger}},
  \bibnamefont{and} \bibinfo{author}{\bibfnamefont{M.}~\bibnamefont{Zukowski}},
  \bibinfo{journal}{Rev.\ Mod.\ Phys.} \textbf{\bibinfo{volume}{84}},
  \bibinfo{pages}{777} (\bibinfo{year}{2012}).

\bibitem[{\citenamefont{Yao et~al.}(2012)\citenamefont{Yao, Wang, Xu, Lu, Pan,
  Bao, Peng, Lu, Chen, and Pan}}]{yao12}
\bibinfo{author}{\bibfnamefont{X.-C.} \bibnamefont{Yao}},
  \bibinfo{author}{\bibfnamefont{T.-X.} \bibnamefont{Wang}},
  \bibinfo{author}{\bibfnamefont{P.}~\bibnamefont{Xu}},
  \bibinfo{author}{\bibfnamefont{H.}~\bibnamefont{Lu}},
  \bibinfo{author}{\bibfnamefont{G.-S.} \bibnamefont{Pan}},
  \bibinfo{author}{\bibfnamefont{X.-H.} \bibnamefont{Bao}},
  \bibinfo{author}{\bibfnamefont{C.-Z.} \bibnamefont{Peng}},
  \bibinfo{author}{\bibfnamefont{C.-Y.} \bibnamefont{Lu}},
  \bibinfo{author}{\bibfnamefont{Y.-A.} \bibnamefont{Chen}}, \bibnamefont{and}
  \bibinfo{author}{\bibfnamefont{J.-W.} \bibnamefont{Pan}},
  \bibinfo{journal}{Nat.\ Photon.} \textbf{\bibinfo{volume}{6}},
  \bibinfo{pages}{225} (\bibinfo{year}{2012}).

\bibitem[{\citenamefont{Mair et~al.}(2001)\citenamefont{Mair, Alipasha, Weihs,
  and Zeilinger}}]{mair01}
\bibinfo{author}{\bibfnamefont{A.}~\bibnamefont{Mair}},
  \bibinfo{author}{\bibfnamefont{V.}~\bibnamefont{Alipasha}},
  \bibinfo{author}{\bibfnamefont{G.}~\bibnamefont{Weihs}}, \bibnamefont{and}
  \bibinfo{author}{\bibfnamefont{A.}~\bibnamefont{Zeilinger}},
  \bibinfo{journal}{Nature} \textbf{\bibinfo{volume}{412}},
  \bibinfo{pages}{313} (\bibinfo{year}{2001}).

\bibitem[{\citenamefont{Barreiro et~al.}(2005)\citenamefont{Barreiro, Langford,
  Peters, and Kwiat}}]{barreiro05}
\bibinfo{author}{\bibfnamefont{J.~T.} \bibnamefont{Barreiro}},
  \bibinfo{author}{\bibfnamefont{N.~K.} \bibnamefont{Langford}},
  \bibinfo{author}{\bibfnamefont{N.~A.} \bibnamefont{Peters}},
  \bibnamefont{and} \bibinfo{author}{\bibfnamefont{P.~G.} \bibnamefont{Kwiat}},
  \bibinfo{journal}{Phys.\ Rev.\ Lett.} \textbf{\bibinfo{volume}{95}},
  \bibinfo{pages}{260501} (\bibinfo{year}{2005}).

\bibitem[{\citenamefont{Molina-Terriza
  et~al.}(2007)\citenamefont{Molina-Terriza, Torres, and Torner}}]{molina07}
\bibinfo{author}{\bibfnamefont{G.}~\bibnamefont{Molina-Terriza}},
  \bibinfo{author}{\bibfnamefont{J.~P.} \bibnamefont{Torres}},
  \bibnamefont{and} \bibinfo{author}{\bibfnamefont{L.}~\bibnamefont{Torner}},
  \bibinfo{journal}{Nat.\ Phys.} \textbf{\bibinfo{volume}{3}},
  \bibinfo{pages}{305} (\bibinfo{year}{2007}).

\bibitem[{\citenamefont{Lanyon et~al.}(2009)\citenamefont{Lanyon, Barbieri,
  Almeida, Jennewein, Ralph, Resch, Pryde, O'Brien, Gilchrist, and
  White}}]{lanyon09}
\bibinfo{author}{\bibfnamefont{B.~P.} \bibnamefont{Lanyon}},
  \bibinfo{author}{\bibfnamefont{M.}~\bibnamefont{Barbieri}},
  \bibinfo{author}{\bibfnamefont{M.~P.} \bibnamefont{Almeida}},
  \bibinfo{author}{\bibfnamefont{T.}~\bibnamefont{Jennewein}},
  \bibinfo{author}{\bibfnamefont{T.~C.} \bibnamefont{Ralph}},
  \bibinfo{author}{\bibfnamefont{K.~J.} \bibnamefont{Resch}},
  \bibinfo{author}{\bibfnamefont{G.~J.} \bibnamefont{Pryde}},
  \bibinfo{author}{\bibfnamefont{J.~L.} \bibnamefont{O'Brien}},
  \bibinfo{author}{\bibfnamefont{A.}~\bibnamefont{Gilchrist}},
  \bibnamefont{and} \bibinfo{author}{\bibfnamefont{A.~G.} \bibnamefont{White}},
  \bibinfo{journal}{Nat.\ Phys.} \textbf{\bibinfo{volume}{5}},
  \bibinfo{pages}{134} (\bibinfo{year}{2009}).

\bibitem[{\citenamefont{Ceccarelli et~al.}(2009)\citenamefont{Ceccarelli,
  Vallone, {De Martini}, Mataloni, and Cabello}}]{ceccarelli09}
\bibinfo{author}{\bibfnamefont{R.}~\bibnamefont{Ceccarelli}},
  \bibinfo{author}{\bibfnamefont{G.}~\bibnamefont{Vallone}},
  \bibinfo{author}{\bibfnamefont{F.}~\bibnamefont{{De Martini}}},
  \bibinfo{author}{\bibfnamefont{P.}~\bibnamefont{Mataloni}}, \bibnamefont{and}
  \bibinfo{author}{\bibfnamefont{A.}~\bibnamefont{Cabello}},
  \bibinfo{journal}{Phys.\ Rev.\ Lett.} \textbf{\bibinfo{volume}{103}},
  \bibinfo{pages}{160401} (\bibinfo{year}{2009}).

\bibitem[{\citenamefont{Nagali et~al.}(2010{\natexlab{a}})\citenamefont{Nagali,
  Sansoni, Marrucci, Santamato, and Sciarrino}}]{nagali10pra}
\bibinfo{author}{\bibfnamefont{E.}~\bibnamefont{Nagali}},
  \bibinfo{author}{\bibfnamefont{L.}~\bibnamefont{Sansoni}},
  \bibinfo{author}{\bibfnamefont{L.}~\bibnamefont{Marrucci}},
  \bibinfo{author}{\bibfnamefont{E.}~\bibnamefont{Santamato}},
  \bibnamefont{and}
  \bibinfo{author}{\bibfnamefont{F.}~\bibnamefont{Sciarrino}},
  \bibinfo{journal}{Phys.\ Rev.\ A} \textbf{\bibinfo{volume}{81}},
  \bibinfo{pages}{052317} (\bibinfo{year}{2010}{\natexlab{a}}).

\bibitem[{\citenamefont{Straupe and Kulik}(2010)}]{straupe10}
\bibinfo{author}{\bibfnamefont{S.}~\bibnamefont{Straupe}} \bibnamefont{and}
  \bibinfo{author}{\bibfnamefont{S.}~\bibnamefont{Kulik}},
  \bibinfo{journal}{Nat.\ Photon.} \textbf{\bibinfo{volume}{4}},
  \bibinfo{pages}{585} (\bibinfo{year}{2010}).

\bibitem[{\citenamefont{Nagali et~al.}(2010{\natexlab{b}})\citenamefont{Nagali,
  Giovannini, Marrucci, Slussarenko, Santamato, and Sciarrino}}]{nagali10prl}
\bibinfo{author}{\bibfnamefont{E.}~\bibnamefont{Nagali}},
  \bibinfo{author}{\bibfnamefont{D.}~\bibnamefont{Giovannini}},
  \bibinfo{author}{\bibfnamefont{L.}~\bibnamefont{Marrucci}},
  \bibinfo{author}{\bibfnamefont{S.}~\bibnamefont{Slussarenko}},
  \bibinfo{author}{\bibfnamefont{E.}~\bibnamefont{Santamato}},
  \bibnamefont{and}
  \bibinfo{author}{\bibfnamefont{F.}~\bibnamefont{Sciarrino}},
  \bibinfo{journal}{Phys.\ Rev.\ Lett.} \textbf{\bibinfo{volume}{105}},
  \bibinfo{pages}{073602} (\bibinfo{year}{2010}{\natexlab{b}}).

\bibitem[{\citenamefont{Gao et~al.}(2010)\citenamefont{Gao, Lu, Yao, Xu,
  G{\"u}hne, Goebel, Chen, Peng, Chen, and Pan}}]{gao10}
\bibinfo{author}{\bibfnamefont{W.-B.} \bibnamefont{Gao}},
  \bibinfo{author}{\bibfnamefont{C.-Y.} \bibnamefont{Lu}},
  \bibinfo{author}{\bibfnamefont{X.-C.} \bibnamefont{Yao}},
  \bibinfo{author}{\bibfnamefont{P.}~\bibnamefont{Xu}},
  \bibinfo{author}{\bibfnamefont{O.}~\bibnamefont{G{\"u}hne}},
  \bibinfo{author}{\bibfnamefont{A.}~\bibnamefont{Goebel}},
  \bibinfo{author}{\bibfnamefont{Y.-A.} \bibnamefont{Chen}},
  \bibinfo{author}{\bibfnamefont{C.-Z.} \bibnamefont{Peng}},
  \bibinfo{author}{\bibfnamefont{Z.-B.} \bibnamefont{Chen}}, \bibnamefont{and}
  \bibinfo{author}{\bibfnamefont{J.-W.} \bibnamefont{Pan}},
  \bibinfo{journal}{Nat.\ Phys.} \textbf{\bibinfo{volume}{6}},
  \bibinfo{pages}{331} (\bibinfo{year}{2010}).

\bibitem[{\citenamefont{Pile}(2011)}]{pile11}
\bibinfo{author}{\bibfnamefont{D.}~\bibnamefont{Pile}}, \bibinfo{journal}{Nat.\
  Photon.} \textbf{\bibinfo{volume}{6}}, \bibinfo{pages}{14}
  (\bibinfo{year}{2011}).

\bibitem[{\citenamefont{Dada et~al.}(2011)\citenamefont{Dada, Leach, Buller,
  Padgett, and Andersson}}]{dada11}
\bibinfo{author}{\bibfnamefont{A.~C.} \bibnamefont{Dada}},
  \bibinfo{author}{\bibfnamefont{J.}~\bibnamefont{Leach}},
  \bibinfo{author}{\bibfnamefont{G.~S.} \bibnamefont{Buller}},
  \bibinfo{author}{\bibfnamefont{M.~J.} \bibnamefont{Padgett}},
  \bibnamefont{and}
  \bibinfo{author}{\bibfnamefont{E.}~\bibnamefont{Andersson}},
  \bibinfo{journal}{Nat.\ Phys.} \textbf{\bibinfo{volume}{7}},
  \bibinfo{pages}{677} (\bibinfo{year}{2011}).

\bibitem[{\citenamefont{Munro et~al.}(2012)\citenamefont{Munro, Stephens,
  Devitt, Harrison, and Nemoto}}]{munro12}
\bibinfo{author}{\bibfnamefont{W.~J.} \bibnamefont{Munro}},
  \bibinfo{author}{\bibfnamefont{A.~M.} \bibnamefont{Stephens}},
  \bibinfo{author}{\bibfnamefont{S.~J.} \bibnamefont{Devitt}},
  \bibinfo{author}{\bibfnamefont{K.~H.} \bibnamefont{Harrison}},
  \bibnamefont{and} \bibinfo{author}{\bibfnamefont{K.}~\bibnamefont{Nemoto}},
  \bibinfo{journal}{Nat.\ Photon.} \textbf{\bibinfo{volume}{6}},
  \bibinfo{pages}{777} (\bibinfo{year}{2012}).

\bibitem[{\citenamefont{Vitelli et~al.}(2013)\citenamefont{Vitelli, Spagnolo,
  Aparo, Sciarrino, Santamato, and Marrucci}}]{vitelli13}
\bibinfo{author}{\bibfnamefont{C.}~\bibnamefont{Vitelli}},
  \bibinfo{author}{\bibfnamefont{N.}~\bibnamefont{Spagnolo}},
  \bibinfo{author}{\bibfnamefont{L.}~\bibnamefont{Aparo}},
  \bibinfo{author}{\bibfnamefont{F.}~\bibnamefont{Sciarrino}},
  \bibinfo{author}{\bibfnamefont{E.}~\bibnamefont{Santamato}},
  \bibnamefont{and} \bibinfo{author}{\bibfnamefont{L.}~\bibnamefont{Marrucci}},
  \bibinfo{journal}{Nat.\ Photon.} \textbf{\bibinfo{volume}{7}},
  \bibinfo{pages}{521} (\bibinfo{year}{2013}).

\bibitem[{\citenamefont{Garc\'{\i}a-Escart\'{\i}n and
  Chamorro-Posada}(2008)}]{garciaescartin08}
\bibinfo{author}{\bibfnamefont{J.~C.} \bibnamefont{Garc\'{\i}a-Escart\'{\i}n}}
  \bibnamefont{and}
  \bibinfo{author}{\bibfnamefont{P.}~\bibnamefont{Chamorro-Posada}},
  \bibinfo{journal}{Phys.\ Rev.\ A} \textbf{\bibinfo{volume}{78}},
  \bibinfo{pages}{062320} (\bibinfo{year}{2008}).

\bibitem[{\citenamefont{Julsgaard et~al.}(2004)\citenamefont{Julsgaard,
  Sherson, Cirac, Fiurasek, and Polzik}}]{julsgaard04}
\bibinfo{author}{\bibfnamefont{B.}~\bibnamefont{Julsgaard}},
  \bibinfo{author}{\bibfnamefont{J.}~\bibnamefont{Sherson}},
  \bibinfo{author}{\bibfnamefont{J.~I.} \bibnamefont{Cirac}},
  \bibinfo{author}{\bibfnamefont{J.}~\bibnamefont{Fiurasek}}, \bibnamefont{and}
  \bibinfo{author}{\bibfnamefont{E.~S.} \bibnamefont{Polzik}},
  \bibinfo{journal}{Nature} \textbf{\bibinfo{volume}{432}},
  \bibinfo{pages}{482} (\bibinfo{year}{2004}).

\bibitem[{\citenamefont{Kimble}(2008)}]{kimble08}
\bibinfo{author}{\bibfnamefont{H.~J.} \bibnamefont{Kimble}},
  \bibinfo{journal}{Nature} \textbf{\bibinfo{volume}{453}},
  \bibinfo{pages}{1023} (\bibinfo{year}{2008}).

\bibitem[{\citenamefont{Grace et~al.}(2006)\citenamefont{Grace, Brif, Rabitz,
  Walmsley, Kosut, and Lidar}}]{grace06}
\bibinfo{author}{\bibfnamefont{M.}~\bibnamefont{Grace}},
  \bibinfo{author}{\bibfnamefont{C.}~\bibnamefont{Brif}},
  \bibinfo{author}{\bibfnamefont{H.}~\bibnamefont{Rabitz}},
  \bibinfo{author}{\bibfnamefont{I.}~\bibnamefont{Walmsley}},
  \bibinfo{author}{\bibfnamefont{R.}~\bibnamefont{Kosut}}, \bibnamefont{and}
  \bibinfo{author}{\bibfnamefont{D.}~\bibnamefont{Lidar}},
  \bibinfo{journal}{New J. Phys.} \textbf{\bibinfo{volume}{8}},
  \bibinfo{pages}{35} (\bibinfo{year}{2006}).

\bibitem[{\citenamefont{Neergaard-Nielsen}(2013)}]{neergaard13}
\bibinfo{author}{\bibfnamefont{J.~S.} \bibnamefont{Neergaard-Nielsen}},
  \bibinfo{journal}{Nat.\ Photon.} \textbf{\bibinfo{volume}{7}},
  \bibinfo{pages}{512} (\bibinfo{year}{2013}).

\bibitem[{\citenamefont{Knill et~al.}(2001)\citenamefont{Knill, Laflamme, and
  Milburn}}]{knill01}
\bibinfo{author}{\bibfnamefont{E.}~\bibnamefont{Knill}},
  \bibinfo{author}{\bibfnamefont{R.}~\bibnamefont{Laflamme}}, \bibnamefont{and}
  \bibinfo{author}{\bibfnamefont{G.~J.} \bibnamefont{Milburn}},
  \bibinfo{journal}{Nature} \textbf{\bibinfo{volume}{409}}, \bibinfo{pages}{46}
  (\bibinfo{year}{2001}).

\bibitem[{\citenamefont{O'Brien et~al.}(2003)\citenamefont{O'Brien, Pryde,
  White, Ralph, and Branning}}]{obrien03}
\bibinfo{author}{\bibfnamefont{J.~L.} \bibnamefont{O'Brien}},
  \bibinfo{author}{\bibfnamefont{G.~J.} \bibnamefont{Pryde}},
  \bibinfo{author}{\bibfnamefont{A.~G.} \bibnamefont{White}},
  \bibinfo{author}{\bibfnamefont{T.~C.} \bibnamefont{Ralph}}, \bibnamefont{and}
  \bibinfo{author}{\bibfnamefont{D.}~\bibnamefont{Branning}},
  \bibinfo{journal}{Nature} \textbf{\bibinfo{volume}{426}},
  \bibinfo{pages}{264} (\bibinfo{year}{2003}).

\bibitem[{\citenamefont{Pittman et~al.}(2003)\citenamefont{Pittman, Fitch,
  Jacobs, and Franson}}]{pittman03}
\bibinfo{author}{\bibfnamefont{T.~B.} \bibnamefont{Pittman}},
  \bibinfo{author}{\bibfnamefont{M.~J.} \bibnamefont{Fitch}},
  \bibinfo{author}{\bibfnamefont{B.~C.} \bibnamefont{Jacobs}},
  \bibnamefont{and} \bibinfo{author}{\bibfnamefont{J.~D.}
  \bibnamefont{Franson}}, \bibinfo{journal}{Phys.\ Rev.\ A}
  \textbf{\bibinfo{volume}{68}}, \bibinfo{pages}{032316}
  (\bibinfo{year}{2003}).

\bibitem[{\citenamefont{Gasparoni et~al.}(2004)\citenamefont{Gasparoni, Pan,
  Walther, Rudolph, and Zeilinger}}]{gasparoni04}
\bibinfo{author}{\bibfnamefont{S.}~\bibnamefont{Gasparoni}},
  \bibinfo{author}{\bibfnamefont{J.-W.} \bibnamefont{Pan}},
  \bibinfo{author}{\bibfnamefont{P.}~\bibnamefont{Walther}},
  \bibinfo{author}{\bibfnamefont{T.}~\bibnamefont{Rudolph}}, \bibnamefont{and}
  \bibinfo{author}{\bibfnamefont{A.}~\bibnamefont{Zeilinger}},
  \bibinfo{journal}{Phys.\ Rev.\ Lett.} \textbf{\bibinfo{volume}{93}},
  \bibinfo{pages}{020504} (\bibinfo{year}{2004}).

\bibitem[{\citenamefont{Zhao et~al.}(2005)\citenamefont{Zhao, Zhang, Chen,
  Zhang, Du, Yang, and Pan}}]{zhao05}
\bibinfo{author}{\bibfnamefont{Z.}~\bibnamefont{Zhao}},
  \bibinfo{author}{\bibfnamefont{A.-N.} \bibnamefont{Zhang}},
  \bibinfo{author}{\bibfnamefont{Y.-A.} \bibnamefont{Chen}},
  \bibinfo{author}{\bibfnamefont{H.}~\bibnamefont{Zhang}},
  \bibinfo{author}{\bibfnamefont{J.-F.} \bibnamefont{Du}},
  \bibinfo{author}{\bibfnamefont{T.}~\bibnamefont{Yang}}, \bibnamefont{and}
  \bibinfo{author}{\bibfnamefont{J.-W.} \bibnamefont{Pan}},
  \bibinfo{journal}{Phys.\ Rev.\ Lett.} \textbf{\bibinfo{volume}{94}},
  \bibinfo{pages}{030501} (\bibinfo{year}{2005}).

\bibitem[{\citenamefont{Pittman et~al.}(2001)\citenamefont{Pittman, Jacobs, and
  Franson}}]{pittman01}
\bibinfo{author}{\bibfnamefont{T.~B.} \bibnamefont{Pittman}},
  \bibinfo{author}{\bibfnamefont{B.~C.} \bibnamefont{Jacobs}},
  \bibnamefont{and} \bibinfo{author}{\bibfnamefont{J.~D.}
  \bibnamefont{Franson}}, \bibinfo{journal}{Phys.\ Rev.\ A}
  \textbf{\bibinfo{volume}{64}}, \bibinfo{pages}{062311}
  (\bibinfo{year}{2001}).

\bibitem[{\citenamefont{Yoran and Reznik}(2003)}]{yoran03}
\bibinfo{author}{\bibfnamefont{N.}~\bibnamefont{Yoran}} \bibnamefont{and}
  \bibinfo{author}{\bibfnamefont{B.}~\bibnamefont{Reznik}},
  \bibinfo{journal}{Phys.\ Rev.\ Lett.} \textbf{\bibinfo{volume}{91}},
  \bibinfo{pages}{037903} (\bibinfo{year}{2003}).

\bibitem[{\citenamefont{Ralph et~al.}(2002)\citenamefont{Ralph, Langford, Bell,
  and White}}]{ralph02}
\bibinfo{author}{\bibfnamefont{T.~C.} \bibnamefont{Ralph}},
  \bibinfo{author}{\bibfnamefont{N.~K.} \bibnamefont{Langford}},
  \bibinfo{author}{\bibfnamefont{T.~B.} \bibnamefont{Bell}}, \bibnamefont{and}
  \bibinfo{author}{\bibfnamefont{A.~G.} \bibnamefont{White}},
  \bibinfo{journal}{Phys.\ Rev.\ A} \textbf{\bibinfo{volume}{65}},
  \bibinfo{pages}{062324} (\bibinfo{year}{2002}).

\bibitem[{\citenamefont{Hofmann and Takeuchi}(2002)}]{hofman02}
\bibinfo{author}{\bibfnamefont{H.~F.} \bibnamefont{Hofmann}} \bibnamefont{and}
  \bibinfo{author}{\bibfnamefont{S.}~\bibnamefont{Takeuchi}},
  \bibinfo{journal}{Phys. Rev. A} \textbf{\bibinfo{volume}{66}},
  \bibinfo{pages}{024308} (\bibinfo{year}{2002}).

\bibitem[{\citenamefont{Ralph}(2004)}]{ralph04}
\bibinfo{author}{\bibfnamefont{T.~C.} \bibnamefont{Ralph}},
  \bibinfo{journal}{Phys.\ Rev.\ A} \textbf{\bibinfo{volume}{70}},
  \bibinfo{pages}{012312} (\bibinfo{year}{2004}).

\bibitem[{\citenamefont{Grudka and W\'ojcik}(2003)}]{grudka03}
\bibinfo{author}{\bibfnamefont{A.}~\bibnamefont{Grudka}} \bibnamefont{and}
  \bibinfo{author}{\bibfnamefont{A.}~\bibnamefont{W\'ojcik}},
  \bibinfo{journal}{Phys.\ Lett.\ A} \textbf{\bibinfo{volume}{314}},
  \bibinfo{pages}{350} (\bibinfo{year}{2003}).

\bibitem[{\citenamefont{Bart\r{u}\u{s}kov\'a
  et~al.}(2006)\citenamefont{Bart\r{u}\u{s}kov\'a, \u{C}ernoch, Filip,
  Fiur\'{a}\u{s}ek, Soubusta, and Du\u{s}ek}}]{bartuskova06}
\bibinfo{author}{\bibfnamefont{L.}~\bibnamefont{Bart\r{u}\u{s}kov\'a}},
  \bibinfo{author}{\bibfnamefont{A.}~\bibnamefont{\u{C}ernoch}},
  \bibinfo{author}{\bibfnamefont{R.}~\bibnamefont{Filip}},
  \bibinfo{author}{\bibfnamefont{J.}~\bibnamefont{Fiur\'{a}\u{s}ek}},
  \bibinfo{author}{\bibfnamefont{J.}~\bibnamefont{Soubusta}}, \bibnamefont{and}
  \bibinfo{author}{\bibfnamefont{M.}~\bibnamefont{Du\u{s}ek}},
  \bibinfo{journal}{Phys.\ Rev.\ A} \textbf{\bibinfo{volume}{74}},
  \bibinfo{pages}{022325} (\bibinfo{year}{2006}).

\bibitem[{\citenamefont{Shahmoon et~al.}(2011)\citenamefont{Shahmoon, Kurizki,
  Fleischhauer, and Petrosyan}}]{shahmoon11}
\bibinfo{author}{\bibfnamefont{E.}~\bibnamefont{Shahmoon}},
  \bibinfo{author}{\bibfnamefont{G.}~\bibnamefont{Kurizki}},
  \bibinfo{author}{\bibfnamefont{M.}~\bibnamefont{Fleischhauer}},
  \bibnamefont{and}
  \bibinfo{author}{\bibfnamefont{D.}~\bibnamefont{Petrosyan}},
  \bibinfo{journal}{Phys.\ Rev.\ A} \textbf{\bibinfo{volume}{83}},
  \bibinfo{pages}{033806} (\bibinfo{year}{2011}).

\bibitem[{\citenamefont{Peyronel et~al.}(2012)\citenamefont{Peyronel,
  Firstenberg, Liang, Hofferberth, Gorshkov, Pohl, Lukin, and
  Vuleti{\'c}}}]{peyronel12}
\bibinfo{author}{\bibfnamefont{T.}~\bibnamefont{Peyronel}},
  \bibinfo{author}{\bibfnamefont{O.}~\bibnamefont{Firstenberg}},
  \bibinfo{author}{\bibfnamefont{Q.-Y.} \bibnamefont{Liang}},
  \bibinfo{author}{\bibfnamefont{S.}~\bibnamefont{Hofferberth}},
  \bibinfo{author}{\bibfnamefont{A.~V.} \bibnamefont{Gorshkov}},
  \bibinfo{author}{\bibfnamefont{T.}~\bibnamefont{Pohl}},
  \bibinfo{author}{\bibfnamefont{M.~D.} \bibnamefont{Lukin}}, \bibnamefont{and}
  \bibinfo{author}{\bibfnamefont{V.}~\bibnamefont{Vuleti{\'c}}},
  \bibinfo{journal}{Nature} \textbf{\bibinfo{volume}{488}}, \bibinfo{pages}{57}
  (\bibinfo{year}{2012}).

\end{thebibliography}

\end{document}